\documentclass[a4paper]{jpconf}

\usepackage{amssymb}
\usepackage{graphicx}
\usepackage{tikz}
\usepackage{amsthm}
\usepackage{amsmath,bm}
\newtheorem{theorem}{Theorem}
\newtheorem{lemma}{Lemma}
\newtheorem{corollary}{Corollary}
\newtheorem{proposition}{Proposition}

\begin{document}
  

\title{Noether's Theorem and its Complement: A Gateway to Particle Interaction}
\author{\bf{Walter Smilga}\footnote{Isardamm 135 d, D-82538 Geretsried, Germany; wsmilga@compuserve.com}
\footnote{Extended version of a talk given at X International Symposium ``Quantum Theory and Symmetries'' 
(QTS-10) in Varna on June 20, 2017.}}

\begin{abstract}
Noether's theorem has gained outstanding importance in theoretical particle physics, because it 
leads to basic conservation laws, such as the conservation of momentum and of angular momentum.
Closely related to this theorem, but unnoticed so far, is a complementary law, which requires the 
(virtual) exchange of momentum between the particles of a closed multi-particle system.
This exchange of momentum determines an interaction. 
For a two-particle system defined by an irreducible representation of the Poincar\'e group, this 
interaction is identified as the electromagnetic interaction. 
This sheds new light on the particle interactions described by the Standard Model.
It resolves long-standing questions about the value of the electromagnetic coupling constant, 
and about divergent integrals in quantum electrodynamics.  
\end{abstract}

\section{Introduction}

While preparing this paper, I found a description of the role of group theory in physics which perfectly
matches the intention of my paper:\\
`Anything that group theory does can be done without it. 
However, not using group theory is like not using a map---you never see the big picture and may go down many blind paths. 
Group theory can give you a lot of information with very little input.' 
({James B. Calvert, Associate Professor Emeritus of Engineering, University of Denver})

I will put this statement to the test on the example of quantum electrodynamics (QED).
Since QED was cast into its present form in 1949/50, we have been faced with two problems. 
The first problem concerns the mathematical inconsistencies of the perturbation algorithm, which become apparent in divergent integrals 
as soon as higher order approximations to the perturbation series are considered. 
Although we have learned to remove these divergences by a mathematical trick called renormalization, this trick has neither made the 
mathematics consistent nor has it contributed to a better understanding of QED. 
The second problem is that QED cannot determine the value of the electromagnetic coupling constant, which indicates that QED in its current 
form is not a closed theory of the electromagnetic interaction, but merely its phenomenological description.

In the absence of a consistent mathematical formulation of QED, interpretations  
of the perturbation algorithm have developed that have more and more been understood as properties of the real world. 
This led David Mermin \cite{dm} to remark `It is a bad habit of physicists to take their most successful abstractions to be real 
properties of our world.' 
Alfred North Whitehead \cite{anw} stated that `reification is part of normal usage of natural language, as well as of literature, where a 
reified abstraction is intended as a figure of speech, and actually understood as such. But the use of reification in logical reasoning or 
rhetoric is misleading and usually regarded as a fallacy.'
The discourse about QED relies heavily on reifications, just to mention `virtual particles', `off-shell behaviour', 
`vacuum polarization', and `vacuum fluctuation'. 
This makes it difficult to gain a clear and unbiased view of QED.

Feynman formulated QED \cite{rf1,rf2,rf3} as an S matrix theory in terms of propagation functions. 
The Feynman rules determine how these functions must be combined with vertex factors to create the S matrix for a given physical problem. 
The rules are visualized by the Feynman graphs, where internal lines correspond to propagation functions, and external lines to the incoming and 
outgoing fermions. 
Feynman, according to his own words, formulated QED as a quantum mechanical `description of a direct action at a distance (albeit delayed in time) 
between charges'.

There is another formulation of QED by Schwinger \cite{js} and Tomonaga \cite{st}, which is based on the concept of fermion fields interacting 
with the quantized electromagnetic field. 
The (local) interaction term of this quantum field theory is determined by the requirements of covariance under Lorentz transformations and of 
`minimal coupling'.
The requirement of minimal coupling was later replaced by the postulate of gauge invariance, which nowadays is upheld as a first principle 
of Nature. 

Dyson \cite{fdy} has shown that the field-theoretical formulation can be understood as an abstraction, from which Feynman's version of QED can 
be obtained by applying the field theory to concrete problems. 
In this sense, both formulations are physically equivalent. 

The fact that Feynman's QED is not formulated as a quantum field theory, but as (relativistic) quantum mechanics, is concealed by the metaphorical 
interpretation of the internal lines of Feynman diagrams as `virtual particles'. 
This interpretation, although a misleading reification, is a popular way of talking about Feynman diagrams without going into the mathematical 
details of the Feynman rules. 

Since particle physics is to a large extent determined by its symmetry with respect to the Poincar\'e group, I have based my paper on group theory, 
which is sufficiently general to avoid undue restrictions and sufficiently abstract to be less susceptible to misleading reifications.
I do not make any assumptions about gauge invariance or the existence of quantum fields, and presume only the validity of the quantum
mechanical axioms and of Poincar\'e invariance; these premises cannot be called into question, since they are experimentally well confirmed.  

The central question to be answered concerns the phenomenon of interaction itself.
The Feynman rules in momentum space \cite{sss} assert that for the construction of the S matrix, a factor ${(2\pi)^4 \, \delta^{(4)}(p-p'+k)}$ 
has to be inserted at each vertex, where $p$ and $ p'$ are the momenta of the fermion lines and $k$ that of a photon line, and that one is to 
integrate over the momenta of all internal lines. 
The integration leads to momentum entangled structures, describing a `virtual exchange of momentum' between the particles, while conserving the 
total momentum. 
An short analysis of S-matrix elements, calculated by Feynman rules, shows that it is this exchange of momentum that is de facto responsible 
for the interaction described by QED, irrespective of what physical or mathematical reason causes the entanglement.

The Standard Model suggests that the reason for entanglement is the exchange of virtual gauge particles, their existence being postulated 
by the principle of gauge invariance.
Considering, however, that the structure not only of single-particle but also of two-particle states is largely determined by their symmetry 
group, i.e. the Poincar\'e group, we can expect that group theory will provide a different, strictly mathematical answer to the question: 
What physical or mathematical conditions can force two single-particle states into a momentum entangled two-particle state?

In answering this question, the group theoretical approach will provide unexpected insights into the quantum mechanics of multi-particle systems 
(`the big picture') and, as by-products, give answers to the issues raised about divergences and coupling constants.

This paper is also intended as a continuation of a previous one \cite{sm1}, which presented a constructive foundation of the axioms of 
quantum mechanics.

\section{The Complement to Noether's Theorem}

Let me start with Noether's theorem, which, on a very basic level, links continuous symmetry groups with conservation laws.
In quantum mechanics, this linkage is especially close, because the generators of unitary symmetry transformations are, at the same time, 
self-adjoint operators that represent observables.
In the case of translation symmetry, the generators of the translations represent the (conserved) momentum; 
in the case of rotational symmetry, the generators of the rotations represent the (conserved) angular momentum.
In the Heisenberg picture, the proof of Noether's theorem is extremely simple: 
the invariance of the Hamiltonian with respect to unitary symmetry transformations means that it commutes with the generator $X$ of the 
symmetry operations. 
According to the Heisenberg equation 
\begin{equation}
\frac{dX}{dt} = i \left[H, X\right],       \label{1-1}
\end{equation}
the self-adjoint operator $X$, now understood as the representation of an observable, is therefore conserved in time. 

In closed multi-particle configurations that are invariant with respect to the operations of the Poincar\'e group, Noether's theorem 
ensures that the total momentum of the system is conserved in time.
If the particles interact, then, in general, this system will not be invariant with respect to translations of a single particle.
In this case, the Heisenberg equation (\ref{1-1}), applied to the generator of such a translation, states that the momentum of this 
particle is not conserved in time.
In general terms, this can be formulated as follows. 
\begin{theorem}[Complement to Noether's theorem]
If the Hamiltonian is not invariant with respect to a continuous unitary 
transformation, then the generator of the transformation is not conserved in time.
\end{theorem} 
\begin{proof} 
The proof follows analogously to the proof of Noether's theorem from the Heisenberg equation.
\end{proof}

Although the mathematical basis of this theorem is the same as that of Noether's original theorem, its physical implications
are quite different:
Because in the abovementioned configuration the total momentum is conserved in time, the change of momentum of one 
particle must be compensated for by a change of momentum of another particle; 
in other words, there is an exchange of momentum between these particles.
Hence, instead of a conservation law, the complement of Noether's theorem causes an interaction law: 
\begin{proposition}[Interaction law]
In a closed multi-particle configuration without translation invariance of the individual particles, the particles exchange momentum.
\end{proposition}
\begin{proof}
The proof follows from applying Noether's theorem to translations of the total system and its complement to translations of 
the individual particles.
\end{proof}

This interaction law can obviously claim the same universal validity as a conservation law derived from Noether's original theorem. 
The question is: Can we find realistic multi-particle systems where the conditions for the application of this law
are met, namely, the existence of continuous unitary transformations and a Hamiltonian that is not invariant under these 
transformations? 
The following two sections will illustrate how these conditions are actually met for a simple two-particle system.

\section{Two-Particle States and Interaction}

According to the axioms of quantum mechanics in combination with Poincar\'e invariance, two independent particles 
with momenta $p_1$ and $p_2$ are described by a product representation of the Poincar\'e group. 
A product representation can be reduced to the direct sum of irreducible representations. 

The irreducible representations of the Poincar\'e group are characterized by fixed
eigenvalues of two Casimir operators \cite{sss1} 
\begin{eqnarray}
P = p^\mu p_\mu\;\; \mbox{and} \;\;
W = -w^\mu w_\mu\; , \; \mbox{with} \;\; 
w_\sigma = \frac{1}{2} \epsilon_{\sigma \mu \nu \lambda} 
M^{\mu \nu}p^\lambda \;.                                             \label{2-1}
\end{eqnarray}
Here, $p^\mu$ and $M^{\mu \nu}$ are the operators of 4-momentum and angular momentum.

The state space $H_I$ of an irreducible representation is a subspace of the state space $H_P$ of the 
corresponding product representation.
In $H_I$ there exists a basis of eigenstates $\left|\mathbf{p},m\right>$ of the total 3-momentum 
$\mathbf{p} = \mathbf{p}_1 + \mathbf{p}_2$ and of a component $m$ of $M^{\mu\nu}$ \cite{sss1}. 
The translations of a single particle, generated by the operators of the individual particle momenta, are 
well-defined unitary transformations within $H_P$, but, in contrast to $H_P$, they are (in general) not symmetry 
transformations of $H_I$. 
In other words, they lead out of $H_I$.
This follows from the commutation relations of the Poincar\'e group \cite{sss1}
\begin{equation}
\left[p^\sigma, M^{\mu\nu}\right] = 
i\,(g^{\mu\sigma} p^\nu - g^{\nu\sigma} p^\mu),        \label{2-3} 
\end{equation} 
which are equal to $0$ only if $\sigma \not= \mu,\nu$. 
If the total momentum $\mathbf{p}$ points in the direction $\sigma$, then $M^{\mu\nu}$ 
does not commute with $\mathbf{p}_1$ and $\mathbf{p}_2$, unless $\mathbf{p}_1$ and 
$\mathbf{p}_2$ are parallel or anti-parallel to $\mathbf{p}$. 
This means the basis states are, in general, not eigenstates of the individual particle momenta:
in consequence, they are not invariant with respect to translations of a single particle. 
Hence, with the exception of the parallel/anti-parallel cases, the following lemma applies.

\begin{lemma} 
Eigenstates of total momentum and orbital angular momentum are not invariant with respect 
to translations of the individual one-particle states.
\end{lemma}

One can say that based on the commutation relations of the Poincar\'e group, the 
conservation of total and angular momentum breaks the translation invariance 
of the individual particles.
 
Lemma 1 largely determines the structure of the basis states $\left|\mathbf{p},m\right>$: 
They are a momentum entangled superposition of product states 
$\left|\mathbf{p_1},\mathbf{p_2}\right>$ with the same total momentum $\mathbf{p}$ 
\begin{equation}
\left|\mathbf{p},m\right> =  
\int_\Omega\!\!d^3\mathbf{p}_1 d^3\mathbf{p}_2\; c(\mathbf{p}, m,\mathbf{p}_1,\mathbf{p}_2) \, 
\left|\,\mathbf{p}_1,\mathbf{p}_2\right>.                          \label{2-4}
\end{equation}    
The coefficients $c(\mathbf{p}, m,\mathbf{p}_1,\mathbf{p}_2)$ are the analogues of the
Clebsch--Gordan coefficients, as known from the coupling of angular momenta.
The domain of integration $\Omega$ is a finite subspace of the two-particle mass shell.
The product states are (in general) momentum entangled, because otherwise the basis states would 
be eigenstates also of the individual particle momenta, which (in general) is excluded by Lemma 1.
The product states are normalized according to
\begin{equation}
\left<\mathbf{p}_1,\mathbf{p}_2 | \mathbf{p}'_1,\mathbf{p}'_2\right> = 
\delta(\mathbf{p}_1 - \mathbf{p}'_1) \, \delta(\mathbf{p}_2 - \mathbf{p}'_2).                      \label{2-6}
\end{equation}
Therefore, the entangled two-particle states (\ref{2-4}) need to be normalized by the factor 
$\omega = V(\Omega)^{-\frac{1}{2}}$, where $V(\Omega)$ is the volume of the domain of integration 
$\Omega$. 
Together, $d^3\mathbf{p}_1 d^3\mathbf{p}_2$ and $\omega$ form an infinitesimal volume element  
that ensures the correct normalization.
For expository reasons, I will not include $\omega$ in the two-particle states, but write 
$\omega \left|\mathbf{p},m\right>$ for the normalized states.

Given this basis, the Hamiltonian of the two-particle system can be written in the form 
\begin{equation}
H = \omega^2 \; \sum_m \int\!d^3\mathbf{p} \; \left|\mathbf{p},m\right> h_{p m} \left<\mathbf{p},m\right|. \nonumber
\end{equation}
Since the basis states are eigenstates of this Hamiltonian, they do not change in time -- except for a phase factor. 
They describe a stable configuration with conserved total and angular momenta: in other words, they describe a closed system.
On the other side, the basis states and, therefore, the Hamiltonian, are not invariant under translations of the 
individual particles.  
Since these translations are well-defined unitary transformations of the underlying product state space, 
Proposition 1 can be applied, leading, together with Lemma 1, to the following statement. 
\begin{corollary} In two-particle systems described by an irreducible representation of the Poincar\'e group, the particles 
exchange (virtual) quanta of momentum.
\end{corollary}
Note that the `exchange of momentum' does not necessarily imply a dynamic process: in the first place, it refers to the 
static entangled structure of two-particle states, which, however, bears the potential for a dynamic process. 
Relating to such structures, Feynman used the phrase `exchange of virtual quanta' \cite{rf2}.

The (virtual or real) exchange of momentum between two particles defines an interaction; the similarity to the electromagnetic 
interaction as described in the Introduction is obvious.
However, this interaction mechanism is not based on a purposely constructed model, as in the case of the Standard Model, but is 
rooted in the principles of quantum mechanics.
The particles interact directly and necessarily, without any mediating particles or fields, thanks to the entangled structure 
of the two-particle state.

Alternatively, Corollary 1 can be obtained directly from the fact that an eigenstate of angular momentum must have a rotational
symmetry, which means that it must be a superposition of product states $\left|\mathbf{p}_1,\mathbf{p}_2 \right>$ such that 
along with any pure product state, the rotated versions of this state also contribute to the eigenstate. 
This necessarily gives the eigenstate a momentum entangled structure.

\section{Illustration: A Scattering Experiment}

The following thought experiment, describing a scattering process, will illustrate the interaction mechanism.

Figure 1 shows an incoming plane wave of particles, some apertures, a target, and a detector.\\ 

\begin{tikzpicture}[scale=0.9, left=0.1cm]
\draw[white] (0.0,0.0)  -- (0.2,0.0);
\draw[rotate=45,black,very thick] (3.017,0.95) -- (3.017,0.05);
\draw[rotate=45,black,very thick] (3.017,-0.15) -- (3.017,-1.05);
\draw[rotate=45,black,very thick] (3,0.95) -- (3,0.05);
\draw[rotate=45,black,very thick] (3,-0.15) -- (3,-1.05);

\draw[black, thin] (3.9,4.1) -- (6.5,4.1);
\draw[black, thin] (3.9,4.1) -- (2.08,2.28);
\draw[rotate=45,black,very thin] (3.3,0.35) -- (3.4,0.15) -- (3.5,0.35);
\draw[rotate=45,black,very thin] (3.3,-0.23) -- (3.4,-0.03) -- (3.5,-0.23);
\draw[rotate=45,black,very thin] (2.0,-0.03) -- (4.0,-0.03);
\draw [black] (2.0,2.1)   node [right,text width=1cm,text centered]{$\mathbf{d}$};

\draw [black] (0.5,0.7)   node [right,text width=1.4cm]{$\mathbf{p}_1$};
\draw [black] (11.08,3.86)   node [right,text width=1.4cm]{$\mathbf{p}_1-\mathbf{k}$};
\draw [black] (3.9,4.50)  node [right,text width=1.4cm]{$\mathbf{p}_2+\mathbf{k}$};

\draw[rotate=0,black,very thick] (6.6,4.90) -- (6.6,4.00);
\draw[rotate=0,black,very thick] (6.6,3.80) -- (6.6,2.90);
\draw[rotate=0,black,very thick] (7.6,4.90) -- (7.6,4.00);

\draw[rotate=0,black,very thick] (7.6,3.80) -- (7.6,2.90);

\draw [black] (2.4,4.50)   node [right,text width=1.5cm]{Target:};
\draw [black] (9.2,3.90)   node [right,text width=2.0cm]{Detector:};

\foreach \x in {10,...,20}
    \draw[rotate=45,black] (\x/7,0.95) -- (\x/7,-1.05);
\filldraw [red] (3.9,4.1) circle (2pt);

\draw [blue, thick, x=0.02cm, y=0.85cm,
declare function={
hyperbel(\t,\a,\b)=sqrt(\t*\t*\a + \b);
}] 
plot [domain=-115:120, samples=50, smooth, rotate=201.5, xshift=-5.15cm, yshift=-2.6cm] (\x - 4.0,{hyperbel(\x/5,0.0031,1.0)}); 
		
\foreach \x in {33,...,45}
    \draw[blue] (\x/5,3.80) -- (\x/5,4.00);

\draw [blue, thick, x=0.02cm, y=0.85cm, declare function={hyperbel(\t,\a,\b)=sqrt(\t*\t*\a + \b);}] 
plot [domain=-115:120, samples=50, smooth, rotate=201.5, xshift=-5.15cm, yshift=-2.6cm] (\x - 4.0,{hyperbel(\x/5,0.0031,1.0)}); 
\end{tikzpicture}
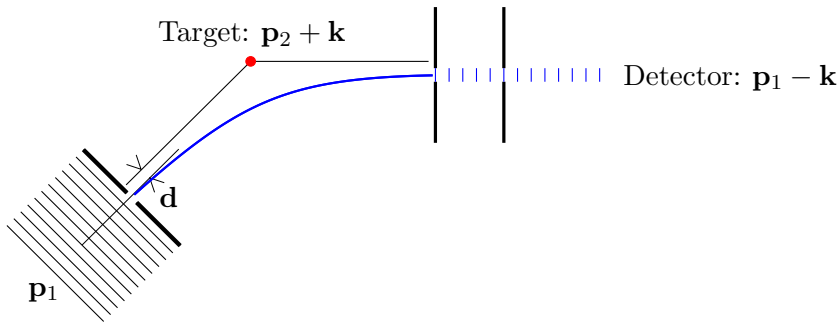
\\ \hspace{2cm}{\bf Figure 1}: Geometry of the scattering experiment\\

Between the first and second aperture, an incoming particle and a particle of the target 
form a two-particle state. 
The apertures of the collimator in front of the detector select an outgoing plane wave. 
The total momentum $\mathbf{p}$ is equal to the incoming momentum $\mathbf{p}_1$. 
In the semi-classical view, $\mathbf{p}_1$, together with the perpendicular distance $\mathbf{d}$ 
between the beam and the target, define an angular momentum $\mathbf{m} = \mathbf{d} \times \mathbf{p}_1$.  
Therefore the experimental setup can be considered a filter that selects intermediate 
eigenstates of angular momentum such as $\omega\left|\mathbf{p},m\right>$. 

Note the similarities with the diffraction of a plane wave at a pin hole: 
both here and there, the basic scattering mechanism becomes visible when essential parts of the incoming plane 
wave are blocked by an aperture. 
Here, the momentum entangled two-particle state is left, there, the spherical elementary wave. 

The scattering amplitude from the incoming product state $\left|\mathbf{p}_1,\mathbf{p}_2 \right>$ 
to the outgoing product state $\left|\mathbf{p}_1-\mathbf{k},\mathbf{p}_2+\mathbf{k}\right>$ 
is given by 
\begin{equation}
S(\mathbf{k}) = \omega^2 \left<\mathbf{p}_1, \mathbf{p}_2|\mathbf{p},m\right>\!
\left<\mathbf{p},m|\mathbf{p}_1-\mathbf{k},\mathbf{p}_2+\mathbf{k}\right>.    \label{3-1}
\end{equation}
Since the intermediate state $\left|\mathbf{p},m\right>$ is momentum entangled, it connects incoming 
and outgoing states also for non-zero values of $\mathbf{k}$. 
There is, in fact, an interaction by an exchange of momentum. 

In the thought experiment, only the value of the Casimir operator $P$ is determined by the 
momentum of the incoming plane wave.
The second Casimir operator $W$ is determined by the geometry of the setup and the sizes of the apertures.

In Equation (\ref{3-1}), the square of the normalization factor $\omega$ of the 
intermediate two-particle state acts like a coupling constant between the incoming and 
the outgoing states. 
This can be formulated as follows.

\begin{corollary} In an irreducible two-particle representation of the Poincar\'e group, 
the square of the normalization factor of a two-particle state defines a coupling 
constant that determines the strength of the interaction between the particles.
\end{corollary}

A comparison of the numerical value of $\omega^2$ with empirical coupling constants 
will help to identify the interaction.

\section{Calculation of the Normalization Factor}

The domain of integration $\Omega$ in the two-particle state (\ref{2-4}) is a finite 
subspace of the two-particle mass shell, parametrized by the momentum vectors 
$\mathbf{p}_1, \mathbf{p}_2$. 
As shown in the following, this mass shell has the topological structure of a fibre 
space. 
It looks very similar to the Hopf fibration \cite{hf} of the hypersphere in four 
dimensions $S^3$.
(On YouTube there is a very instructive visualization of the Hopf fibration 
\cite{nj}.) 

The construction of the infinitesimal volume element on $\Omega$ was described in \cite{sm3}; 
here I provide an improved and simplified version of that construction.

Let $p_1$ and $p_2$ be the 4-momenta of two particles with masses $m_1$ and $m_2$.
They satisfy the mass shell relations
\begin{equation}
{p_1}^2 = m_1^2 \;\; \mbox{ and } \;\; {p_2}^2 = m_2^2. \label{4-1}
\end{equation}               
The total momentum $p$ and the relative momentum $q$ are defined by
\begin{equation}
p = p_1 + p_2 \;\; \mbox{ and } \;\; q = p_1 - p_2   \label{4-2}
\end{equation}
and satisfy
\begin{equation}
p\,q = m_1^2 - m_2^2. \label{4-3}
\end{equation}
In the rest frame of $p$, Equation (\ref{4-3}) allows rotating the vector $q$ relatively to $p$ by 
the action of SO(3), as long as no other restrictions apply.

For an irreducible two-particle representation, the relations
\begin{equation}
p^2 = m_{tot}^2                \label{4-4}   
\end{equation}
and
\begin{equation}
q^2 = 2 m_1^2 + 2 m_2^2 - m_{tot}^2      \label{4-5}
\end{equation}
hold, where the constant $m_{tot}$ is the effective mass of the two-particle system, 
corresponding to the Casimir operator $P \equiv p^2$.
Equations (\ref{4-4}) and (\ref{4-5}) can be combined to form the equation of the 
two-particle mass shell
\begin{equation}
p^2 + q^2 =  2 m_1^2 + 2 m_2^2.         \label{4-6}   
\end{equation}

Rotations around the total momentum $\mathbf{p}$ as the axis leave $p$ invariant but change $q$. 
Therefore, these rotations define an internal degree of freedom with an SO(2) symmetry.
The action of SO(3,1) on $p$, together with the action of SO(2) on $q$,
generates the two-particle mass shell (\ref{4-6}), which parametrizes the state 
space of an irreducible two-particle representation. 
The SO(2) moves within SO(2,1) as $p$ moves through the hyperboloid (\ref{4-4}).
The two-particle mass shell has therefore the structure of a circle bundle over 
an hyperboloid. 
The circle fibres parametrize the internal rotational degree of freedom of a two-particle state 
of an irreducible representation.
The entanglement of the two-particle state (\ref{2-4}) is generated by integration over this 
one-dimensional parameter.


The groups SO(3,1), acting on $p$, and SO(2,1), acting on $q$, are subgroups of 
the group SO(5,2).
(The form of Equation (\ref{4-6}) demonstrates these symmetries.)
This fact can be used to simplify the calculation of the volume element on $\Omega$ 
by starting from a fully SO(5,2) symmetric space, isomorphic to the symmetric space 
SO(5,2)/(SO(5)$\times$SO(2)). 
This is the fourth symmetric domain of Cartan, also known as the Lie ball, 
in 5 dimensions \cite{eca,hua}.

An integral on the Lie ball can be split into a spherical integral over
the boundary $Q^5$ of the unit Lie ball $D^5$ 
\begin{equation}
\int_{Q^5} d^4x \, ...\, ,  \label{4-10}
\end{equation}
and a second integral over the 
radial direction of $D^5$
\begin{equation}
\int dr \, ...\, .                               \label{4-9}
\end{equation}
The spherical integral (\ref{4-10}) is normalized by the inverse of the volume $V\!(Q^5)$. 
This factor is the first contribution to the squared normalization factor $\omega^2$.

For a given radius, $D^5$ can be understood as generated by the action of SO(5), 
whereas the domain $\Omega$ can be understood as generated by the combined action of SO(3) 
and SO(2), equivalent to the action of four rotations with orthogonal axes.
This corresponds to a local SO(4). 
Hence, in comparison to $D^5$, on $\Omega$, a given volume element contains (parametrizes) fewer 
product states contributing to a two-particle state.
The ratio is given by the volume of the symmetric space SO(5)/SO(4), which is the unit sphere 
$S^4$ in 5 dimensions.
Therefore, the volume element on $\Omega$, i.e. $\omega^2$, must be divided by the volume 
$V\!(S^4)$.

The infinitesimal volume element is still a spherical volume element with equal 
(infinitesimal) sizes in the four directions on $Q^5$, but with a different size 
in the radial direction.
To give the infinitesimal volume element the form of an isotropic Cartesian volume 
element, as required in the integral of the two-particle state (\ref{2-4}), the 
size in the radial direction must be adjusted.
On the unit Lie ball, the integration in the radial direction is given by the integral
(\ref{4-9}) with boundaries $0$ and $1$.
The integration of the infinitesimal volume element $dr$, considered as part of 
a Cartesian volume element, contributes a factor $1$ to the integral.
(The normalization factor for the radial direction is therefore equal to 1.)
Together with the integration over the surface of $Q^5$, 
the infinitesimal volume elements add up to the volume of the unit Lie ball $V\!(D^5)$.
Therefore, each of the four elements of $d^4x$, again considered as part of a 
Cartesian volume element, contributes a factor of $V\!(D^5)^{\frac{1}{4}}$ to the 
$V\!(D^5)$.
To make the infinitesimal volume element an isotropic one, $dr$ must be rescaled
by this factor.

The Jacobian that relates $p$ and $q$ to $p_1$ and $p_2$ contributes an additional 
factor of 2.

Taking all factors together results in the squared normalization factor
\begin{equation}
\omega^2 = 2\,V\!(D^5)^{\frac{1}{4}} \, / \, (V\!(Q^5)\,V\!(S^4)).    \label{4-12}
\end{equation}

Finally, irreducible representations with different values of the Casimir operator $W$ 
may contribute additively (and coherently) to the transition amplitude (\ref{3-1}). 
Their contribution derives from comparing the range of $q$ within a product representation, 
as covered by the action of SO(3) (cf. Equation (\ref{4-3})), with the range within an 
irreducible representation, as covered by the action of SO(2).
These contributions increase the transition amplitude by the factor $V$(SO(3)/SO(2)) = $V(S^2)$ = $4\pi$.
Accordingly, the coupling constant in the transition amplitude (\ref{3-1}) must be adjusted
by the factor $4\pi$:
\begin{equation}
4\pi \omega^2 = 8\pi\,V\!(D^5)^{\frac{1}{4}} \, / \, (V\!(Q^5)\,V\!(S^4)).    \label{4-13}
\end{equation}

This value has to be compared with empirical coupling constants. 
Inserting the explicit values (taken from \cite{hua})
\begin{equation}
V\!(D^5) = \frac{\pi^5}{2^4\, 5!},\;\;    	 
V\!(Q^5) = \frac{8 \pi^3}{3},\;\;         	 
V\!(S^4) = \frac{8 \pi^2}{3}                                               	  \label{4-15}
\end{equation}
leads to 
\begin{equation}
4\pi \omega^2 =  \frac{9}{16 \pi^3} \left(\frac{\pi}{120}\right)^{1/4} 
=  1/137.03608245,   			          	\label{4-16}
\vspace{0.2cm}
\end{equation}
which closely matches the CODATA value $1/137.035999139$ \cite{cod} of the 
electromagnetic fine-structure constant $\alpha$.
This agreement clearly identifies the interaction as the electromagnetic interaction. 
Therefore:
\begin{corollary} The electromagnetic interaction is a model-independent property 
of irreducible two-particle representations of the Poincar\'e group.
\end{corollary}

The above calculation also gives $\alpha$ a clear mathematical meaning: it is essentially the square of 
the correct normalization factor for two-particle eigenstates of total momentum and angular momentum.

The volume term on the right side of Equation (\ref{4-13}) was accidentally found in 1971 
by the Swiss mathematician Armand Wyler (a former student of Heinz Hopf), while he was 
playing around with some symmetric domains. 
Wyler published his finding \cite{aw}, hoping to attract the interest of physicists. 
Unfortunately, Wyler was not able to put his observation into a convincing physical 
context. See also \cite{wew}.
Therefore, the physics community dismissed his formula as meaningless numerology. 
Now it has become evident that Wyler's formula directly links the signature of the 
electromagnetic interaction with the geometric footprint of the fibered two-particle 
mass shell.

\section{Electric Charge}

From the above considerations it is obvious that the numerical value of the `electric charge' is nothing 
other than the normalization factor of two-particle states. 
Therefore, charge cannot be considered a property of a single particle: it is rather an emerging property 
of two-particle states.
For massive spin-$1/2$ particles the sign of charge is determined by the Dirac equation, giving particles 
and anti-particles opposite signs.

We can further conclude that two particles in a pure product state do not interact electromagnetically, 
because such a state is by definition not momentum entangled: therefore, there is no exchange of momentum.
This may explain why neutrinos empirically don't show an electric charge.
A suitable theory of weak interaction should clarify this issue.

\section{Electromagnetic Field and Coulomb potential}

Within the perturbation approach to QED it is shown that if one of the particle masses is much larger than the other,
the exchange of virtual photons is equivalent to a $1/r$ potential, which is identified as the Coulomb potential 
\cite{bn,sss3}. 

In \cite{cdg} we find a more direct derivation, which matches our view of momentum exchange: 
it is based on treating the exchanged momenta as a perturbation to a two-particle system with the particles pinned 
to fixed positions.

In quantum mechanics, a momentum $\bm{k}$ is represented by the wave function $e^{i\bm{k}\bm{x}}$.
Based on this representation, the field $A_s$ is defined in terms of `scalar photons'
\begin{equation}
A_s(\bm{x}) = (2\pi)^{-3/2} \int \frac{d^3k}{\sqrt{2\,|\bm{k}|}}\, (a_s(\bm{k}) \, e^{i\bm{k}\bm{x}} 
+ a_s^\dagger(\bm{k}) \, e^{-i\bm{k}\bm{x}})
\end{equation}
describing the exchanged momenta in a Fock space representation. 
It can be understood as the time component $A_0$ of a 4-potential $A_\mu$, in analogy to the quantized electromagnetic 
4-potential.
Following the perturbation approach of QED, the coupling Hamiltonian is defined by
\begin{equation}
H_{coupling} = \int d^3 x \, A_\mu \, j^\mu ,
\end{equation}
where $j^\mu = (\rho, \bm{j})$ is the 4-current density.
In our case of pinned charges it is given by
\begin{equation}
j^\mu = (q_1\,\delta({\bm{x} - \bm{x}_1}) + q_2\,\delta({\bm{x} - \bm{x}_2}), \; \bm{0}),
\end{equation}
which leads to
\begin{equation}
H_{coupling} = (2\pi)^{-3/2} \int \frac{d^3k}{\sqrt{2\,|\bm{k}|}}\,
\left[a_s(\bm{k}) \, (q_1\,e^{i\bm{k}\bm{x}_1} + q_2\,e^{i\bm{k}\bm{x}_2}) + a_s^\dagger(\bm{k}) \, 
(q_1\,e^{-i\bm{k}\bm{x}_1} + q_2\,e^{-i\bm{k}\bm{x}_2})\right].
\end{equation}

We can estimate the shift $\Delta E$ by which the energy is changed by using perturbation theory 
up to the first non-zero term:
\begin{equation}
\Delta E = \left<0\right|H_{coupling}\left|0\right> 
+ \sum_{n \not= 0} \frac{|\left<n\right|H_{coupling}\left|0\right>|^2}{E_n} + ...\,.  \label{7-2}
\end{equation}
The first term of (\ref{7-2}) is zero; the second term is non-zero for states where a scalar photon is produced.
With $E_n \sim |\bm{k}|$, we thus get 
\begin{equation}
\Delta E \approx \int d^3k \, \frac{|\left<\bm{k}\right|H_{coupling}\left|0\right>|^2}{|\bm{k}|}.  \label{7-3}
\end{equation}
We now have
\begin{equation}
\left<\bm{k}\right|H_{coupling}\left|0\right> = \frac{1}{\sqrt{2\, |\bm{k}|}} \, (q_1\,e^{i\bm{k}\bm{x}_1} 
+ q_2\,e^{i\bm{k}\bm{x}_2}) 
\end{equation}
and, therefore,
\begin{equation}
|\left<\bm{k}\right|H_{coupling}\left|0\right>|^2 = \frac{1}{2\, |\bm{k}|} \, |(q_1\,e^{i\bm{k}\bm{x}_1} 
+ q_2\,e^{i\bm{k}\bm{x}_2})|^2.   \label{7-5}
\end{equation}
Substituting (\ref{7-5}) into (\ref{7-3}) gives 
\begin{equation}
\Delta E \approx \epsilon_1 + \epsilon_2 + V_{coul},
\end{equation}
where 
\begin{equation}
\epsilon =  \int d^3k\, \frac{q^2}{2\,k^2}
\end{equation}
is the Coulomb self-energy of the two charges, $q_1$ and $q_2$, and
\begin{equation}
V_{coul} =  \int d^3k\,\frac{q_1q_2\,e^{i\bm{k}(\bm{x}_1-\bm{x}_2)}}{2\,k^2} =  \frac{q_1q_2}{4\pi\,|\bm{x}_1 - \bm{x}_2|}
\end{equation}
the Coulomb interaction energy between them.

\section{Conclusions}

Noether's theorem in connection with its complement provides a new and unbiased view on 
the electromagnetic interaction; 
this view is mathematically well founded on the principles of quantum mechanics, is independent 
of any model, and is physically supported by the matching of the calculated and empirical values of the 
fine-structure constant $\alpha$. 
It simply says: In a closed multi-particle system with well-defined and conserved total 
and angular momenta, the translation invariance of the individual particles is broken; 
therefore, the corresponding multi-particle states cannot be plain product states, but must be 
momentum entangled; 
this is synonymous with a virtual exchange of momentum between the individual particle states. 

In contrast to the standard formulation of QED, the group theoretical approach uniquely determines the 
electromagnetic coupling constant, which is identified as the normalization factor for two-particle 
states of an irreducible two-particle representation of the Poincar\'e group.

Noether's theorem basically confirms the interaction mechanism of the Standard Model, which is an 
interaction by exchange of momentum. 
However, in the Standard Model, the exchange of momentum is modelled by the exchange of virtual gauge bosons; 
this approach does not take note of the special topology of the fibered two-particle mass shell: 
In a correctly defined two-particle state, the exchange of momentum is controlled by a one-dimensional 
and bounded parameter on a circle fibre.
In contrast, gauge particles come along with three independent components of momentum.
The Feynman rules prescribe integration over these three (unbounded) parameters, rather 
than---as would be correct---over a single parameter on a circle fibre.
An inspection of the integrals of the standard perturbation algorithm (cf., e.g. \cite{sss2}) clearly 
shows that it is the excessive number and ranges of integration variables that are responsible for the 
well-known divergences.

The crucial insight of the foregoing analysis is that the structure of the Poincar\'e group, via its irreducible 
representations, completely determines not only the basic properties of single particles, but also their 
(electromagnetic) interaction. 
There is no need for an interaction term or additional physical principles, such as gauge invariance, 
which the developers of the Standard Model considered an indispensable prerequisite of interaction.
Poincar\'e invariance alone provides the rules for the complete description of multi-particle configurations.
The Standard Model implements these rules only insufficiently: 
although the Feynman rules correctly ensure the conservation of the total momentum by explicitly providing 
the appropriate $\delta$ functions, they fail to provide adequate rules for the angular momentum.

Group theory provides a simple and unspectacular `big picture' of the real world, showing nothing more than 
the Hilbert space of a product representation of the Poincar\'e group---this is in full compliance with the 
quantum mechanical axiom concerning composed systems. 
Within this Hilbert space, we find entangled structures, which can be described by a virtual exchange of momentum.
These structures are a consequence of the basic structure of the Poincar\'e group, rather than of an exchange of 
virtual `force particles' (gauge bosons).
The popular picture of `virtual particles constantly popping in and out of existence' \cite{sca}, is thereby 
exposed as a `fallacy of misplaced concreteness' \cite{anw}.
The big picture presents us with the challenge to comprehend also the weak, strong, and, finally, gravitational 
interactions, as inherent structural properties of multi-particle states -- instead of merely modelling them.



\end{document}